\documentclass[final,11pt]{elsarticle} 

\usepackage{verbatim,a4wide}

\usepackage{amsmath}
\usepackage{amssymb}
\usepackage{amsthm}

\usepackage{algorithm}
\usepackage{algpseudocode}

\usepackage{xcolor}

\usepackage{tikz}
\usetikzlibrary{patterns}

\usepackage{hyperref}

\usepackage[cp1250]{inputenc}
\usepackage[OT4]{fontenc}
\usepackage[english]{babel}

\usepackage{datetime}

\numberwithin{equation}{section}
\numberwithin{algorithm}{section}
\numberwithin{figure}{section}
\numberwithin{table}{section}

\newtheorem{theorem}{Theorem}[section]

\newtheorem{remark}[theorem]{Remark}
\newtheorem{corollary}[theorem]{Corollary}
\newtheorem{example}[theorem]{Example}

\newtheorem{problem}[theorem]{Problem}

\newcommand{\p}[1]{\mbox{\textsf{#1}}}

\journal{Numerical Algorithms}

\begin{document}

\begin{frontmatter}

\title{Fast subdivision of B\'{e}zier curves}

\author[A1]{Pawe{\l} Wo\'{z}ny\corref{cor}}
\ead{Pawel.Wozny@cs.uni.wroc.pl}

\author[A1]{Filip Chudy}
\ead{Filip.Chudy@cs.uni.wroc.pl}

\cortext[cor]{Corresponding author.}

\address[A1]{Institute of Computer Science, University of Wroc{\l}aw,
             ul.~Joliot-Curie 15, 50-383 Wroc{\l}aw, Poland}

\begin{abstract}
It is well-known that a $d$-dimensional polynomial B\'{e}zier curve of degree
$n$ can be subdivided into two segments using the famous de Casteljau algorithm
in $O(dn^2)$ time. We show that it is possible to do this in $O(dn\log{n})$ 
time using the fast Fourier transform and its inverse. The direct application 
of the method is numerically unstable for larger values of $n$. Therefore, we 
also introduce a modified version based on parameter scaling, which retains the 
same asymptotic computational complexity. This scaling approach can improve 
numerical accuracy in some cases (for $n\lesssim 50$), as demonstrated by our 
experiments. Moreover, the new method has a nice property: if a B\'{e}zier 
curve is extended by an additional control point, the subdivision can be 
updated in $O(d)$ time. The method can be adapted to subdivide rational 
B\'{e}zier curves and rectangular B\'{e}zier surfaces, as well as to compute 
the derivatives of B\'{e}zier curves.

Although motivated by applications in CAGD, the proposed methods may also be 
useful in other areas involving polynomial reparameterization.
\end{abstract}


\end{frontmatter}

\section{Introduction}                                  \label{S:Introduction}

Let $\p{P}_n:[0,1]\rightarrow{\mathbb E}^d$ $(d\in\mathbb N)$ be a 
(polynomial) \textit{B\'{e}zier curve of degree $n\in\mathbb N$} with the 
\textit{control points} $\p{W}_0,\p{W}_1,\ldots,\p{W}_n\in{\mathbb E}^d$,
\begin{equation}\label{E:BezierCurve}
\p{P}_n(t):=\sum_{i=0}^{n}B^n_i(t)\p{W}_i\qquad (0\leq t\leq 1).
\end{equation}
Here $B^n_i$ denotes the \textit{$i$th Bernstein (basis) polynomial of degree 
$n$},
\begin{equation}\label{E:Def_BernPoly}
B^n_i(t):=\binom{n}{i} t^i (1-t)^{n-i}\qquad (0\leq i\leq n;\; n\in\mathbb N).
\end{equation}

The parametric B\'{e}zier curves have many well-known advantages and are 
therefore among the most popular tools used in computer graphics and CAD/CAGD. 
The standard way to evaluate $\p{P}_n(c)$ for a given $c\in[0,1]$ is to use 
the \textit{de Casteljau algorithm}:
\begin{equation}\label{E:DC-alg}
\left\{
\begin{array}{l}
\p{W}_i^{(0)}:=\p{W}_i\qquad (0\leq i\leq n),\\[1ex]
\p{W}^{(k)}_i:=(1-c)\p{W}^{(k-1)}_i+c\,\p{W}^{(k-1)}_{i+1}\qquad 
                                          (1\leq k\leq n;\; 0\leq i\leq n-k).
\end{array}
\right.
\end{equation}
Then, $\p{P}_n(c)=\p{W}^{(n)}_0$. The algorithm follows from the identity
\begin{equation*}
B^n_k(t)=tB^{n-1}_{k-1}(t)+(1-t)B^{n-1}_k(t)\qquad
                  (0\leq k\leq n;\ \mbox{$B^p_q\equiv 0$ if $q<0$ or $q>p$}).
\end{equation*}

The de Casteljau method is very simple. Moreover, it has an elegant geometric 
interpretation, good numerical properties, and involves only convex 
combinations of points in ${\mathbb E}^d$. However, its drawback is the 
$O(dn^2)$ computational complexity.

\begin{table*}[ht!]
\begin{center}
\renewcommand{\arraystretch}{1.5}
\begin{tabular}{lllll}
$\p{W}^{(0)}_0$\\
$\p{W}^{(0)}_1$ & $\p{W}^{(1)}_0$\\
$\p{W}^{(0)}_2$ & $\p{W}^{(1)}_1$ & $\p{W}^{(2)}_0$\\
\multicolumn{1}{c}{$\vdots$} & \multicolumn{1}{c}{$\vdots$} & 
                                  \multicolumn{1}{c}{$\vdots$} & $\ddots$\\ 
$\p{W}^{(0)}_n$ & $\p{W}^{(1)}_{n-1}$ & $\p{W}^{(2)}_{n-2}$ & 
                           \multicolumn{1}{c}{$\cdots$} & $\p{W}^{(n)}_{0}$\\
\end{tabular}
\end{center}
\renewcommand{\arraystretch}{1}
\caption{The de Casteljau table.}\label{T:DC-table}
\end{table*}

For more properties and applications, as well as the history and 
generalizations of Bernstein polynomials, B\'{e}zier curves, and the de 
Casteljau algorithms (including the \textit{rational case}), see, e.g., 
\cite{Farin2002}, \cite{Farouki2012}. 

Note that the Bernstein basis provides a well-conditioned representation of 
polynomials (see~\cite{Farouki1987}). Algorithms for polynomials in the 
Bernstein form are discussed, e.g., in~\cite{Farouki1988}, while the stability 
of Bernstein-B\'{e}zier subdivision is addressed in~\cite{Farouki1990}.

It is worth mentioning that new linear-time geometric algorithms for evaluating 
polynomial and rational B\'{e}zier curves have been recently presented by the
authors in~\cite{WCh2020}. These methods also have a geometric interpretation,
operate on convex combinations of points and are numerically stable. Although 
their justification is more complex, they are significantly faster than the 
polynomial and rational de Casteljau algorithms, as they work in $O(dn)$ time.

For details, efficient implementations, numerical properties, as well as test 
results, see \cite{WCh2020}, \cite[Chapter 2]{FChPhD}, and \cite{CFudaPhD}, 
\cite{CFuda2023}, \cite{VH2026}.

Note that in the special case $d=1$, all the aforementioned properties, 
methods, and algorithms can be applied to polynomials expressed in the 
Bernstein form, which provides a well-conditioned representation, as mentioned 
above.

\section{The problem}                                        \label{S:Problem}

The de Casteljau algorithm can also be used to \textit{subdivide} a B\'{e}zier 
curve~\eqref{E:BezierCurve} into two segments. This is particularly significant 
for designers working with CAD systems.

For example, if we fix $c\in(0,1)$, then the \textit{left part} of the curve 
$\p{P}_n$ (cf.~\eqref{E:BezierCurve}), i.e., the B\'{e}zier curve
$$
\p{P}^{L}_n([0,1]):=\p{P}_n([0,c])
$$
is given by
\begin{equation}\label{E:Subdivision}
\p{P}^{L}_n(t)=\sum_{i=0}^{n}B^n_i(t)\p{V}_i\qquad (0\leq t\leq 1),
\end{equation}
where the control points $\p{V}_0,\p{V}_1,\ldots,\p{V}_n\in{\mathbb E}^d$ 
are defined as
\begin{equation}\label{E:PointsV}
\p{V}_k\equiv\p{V}_k(c):=\sum_{i=0}^kB^k_i(c)\p{W}_i\qquad (0\leq k\leq n).
\end{equation}
See, e.g., \cite[Eqs.~(5.29), (6.22)]{Farin2002}. 

The points in the de Casteljau table (cf.~Table~\ref{T:DC-table}) satisfy the
following relation:
$$
\p{W}^{(k)}_i\equiv\p{W}^{(k)}_i(c)=\sum_{j=0}^kB^k_j(c)\p{W}_{i+j},
$$
where $0\leq k\leq n$ and $0\leq i\leq n-k$ (cf.~\eqref{E:DC-alg}). For 
details, see, e.g., \cite[Eq.~(5.7)]{Farin2002}.

Note that for $k=0,1,\ldots,n$, $\p{V}_k=\p{W}^{(k)}_0$ 
(cf.~\eqref{E:PointsV}), i.e., the control points of the left part of $\p{P}_n$ 
lie along the diagonal of the de Casteljau table. Similarly, the control points 
of the \textit{right part} of $\p{P}_n$ comprise the bottom row of 
Table~\ref{T:DC-table}. This means that the problems of finding the left and 
the right parts can be easily reduced to each other.

Additionally, by executing the de Casteljau algorithm once, we subdivide 
a $d$-dimensional polynomial B\'{e}zier curve of degree $n$ into two segments 
in $O(dn^2)$ time.

Note that, using the new geometric linear-time method mentioned in 
Section~\ref{S:Introduction}, one can also subdivide a polynomial B\'{e}zier 
curve into two parts. The computational complexity, however, remains $O(dn^2)$. 
See~\cite[p.~3]{WCh2020}. So far, no algorithms for subdividing a polynomial 
B\'{e}zier curve faster than in $O(dn^2)$ time have been proposed.

Accordingly, this paper addresses the following main problem.

\begin{problem}\label{P:MainProblem}
Given $c\in(0,1)$, a natural number $n$ and the points 
$\p{W}_0,\p{W}_1,\ldots,\p{W}_n\in{\mathbb E}^d$, compute all points
$\p{V}_0,\p{V}_1,\ldots,\p{V}_n\in{\mathbb E}^d$
(cf.~\eqref{E:Subdivision} and~\eqref{E:PointsV}) faster than in $O(dn^2)$ 
time.
\end{problem}

The article is organized as follows. In Section~\ref{S:T-Solution}, we show 
that Problem~\ref{P:MainProblem} can be solved in $O(dn\log{n})$ time using 
the \textit{fast Fourier transform} (the FFT in short) and its \textit{inverse} 
(the IFFT). 

Note that the new method has an interesting property. If an already subdivided 
B\'{e}zier curve's control point sequence is extended by an additional point, 
it is enough to find a subdivision of the new curve using only $O(d)$ 
additional arithmetic operations. In the case of the de Casteljau algorithm, 
one would need $O(dn)$ operations.

The solution proposed in Section~\ref{S:T-Solution} is mainly theoretical
because its direct implementation loses numerical stability even for moderate
values of $n$. Therefore, in Section~\ref{S:P-Solution}, we propose a slightly 
modified algorithm that retains the same asymptotic computational complexity 
while potentially improving numerical robustness, making it effective for 
moderately large values of $n$ in some cases (for values of $n$ up to 
approximately 50). We also discuss the difficulty of finding an appropriate 
scaling factor in the general case.

In Section~\ref{S:OtherApp}, we present an outline of how to apply
the method shown in preceding sections to subdivide additional types of 
B\'{e}zier-type objects---rational B\'{e}zier curves and rectangular B\'{e}zier 
surfaces---as well as to compute the derivatives of polynomial B\'{e}zier 
curves. This results in methods which have lower asymptotic computational 
complexity than the traditional de Casteljau algorithm.

The results of some illustrative numerical experiments are presented in 
Section~\ref{S:Examples}, where we also propose a \textit{hybrid method} for 
subdividing B\'{e}zier curves.

We emphasize that, although convolution-based formulations of some basic 
operations on polynomials in Bernstein form are already known 
(see~\cite{Sanchez2003}, where multiplication, degree elevation, composition, 
and differentiation in the scaled Bernstein basis were considered), the present 
paper is, to the best of our knowledge, the first to introduce and exploit 
FFT-based convolution for computing the control points of the curve's 
subdivision.

We conclude this section with two important remarks. Given the good 
conditioning of the Bernstein representation of polynomials 
(see~\cite{Farouki1987}), the results presented in this paper may be of 
interest not only to the CAGD community but also to researchers developing 
methods and algorithms involving, for example, reparameterizations of 
polynomials. Nevertheless, we present the problem and its solution in 
a geometric setting, as the problem itself originates from applications in 
CAGD.

\section{The theoretical solution}                        \label{S:T-Solution}

Recall that our goal is to compute all the points~\eqref{E:PointsV} with
lower asymptotic computational complexity than that of the de Casteljau 
algorithm. We now state the main theoretical result of the paper.

\begin{theorem}\label{T:MainThm}
Problem~\ref{P:MainProblem} can be solved in $O(dn\log{n})$ time.
\end{theorem}
\begin{proof}
Let us consider the case of $d=1$. If $d>1$, we can proceed in 
a component-wise manner.

For a fixed $n\in\mathbb N$, we have $v_k$ as the one-dimensional version of 
$\p{V}_k$: 
\begin{equation}\label{E:Def-v_k}
v_k\equiv v_k(c)=k!\sum_{i=0}^k\frac{w_ic^i}{i!}\cdot
                                     \frac{(1-c)^{k-i}}{(k-i)!}
                                                  \qquad(k=0,1,\ldots,n)
\end{equation}
(cf.~\eqref{E:PointsV}, \eqref{E:Def_BernPoly}), where $c\in(0,1)$, and 
$w_i\in\mathbb R$ $(0\leq i\leq n)$ are the \textit{scalar} control points.

Let us define the following two polynomials of degree $n$ in the variable $x$:
\begin{equation}\label{E:Def-a_n-b_n}
a_n(x):=\sum_{i=0}^{n}\alpha_ix^i,\qquad 
b_n(x):=\sum_{j=0}^{n}\beta_jx^j,
\end{equation}
where
\begin{equation}\label{E:Def-alpha_i-beta_i}
\alpha_i:=\frac{w_i}{i!}c^i,\qquad 
\beta_j:=\frac{1}{j!}(1-c)^j\qquad (0\leq i,j\leq n).
\end{equation}
Define $g_n(x):=a_n(x)b_n(x)$. Certainly, the power form of the 
polynomial $g_n$ is 
\begin{equation}\label{E:Def-g_n-poly}
g_n(x)=\sum_{k=0}^{2n}\gamma_kx^k
\end{equation}
with
\begin{equation}\label{E:Def-gamma_k}
\gamma_k:=\sum_{i=0}^{k}\alpha_i\cdot \beta_{k-i}\qquad (0\leq k\leq 2n),
\end{equation}
where we set $\alpha_\ell:=0$, $\beta_\ell:=0$ if $\ell>n$.

Now, observe that
\begin{equation}\label{E:FinalStep}
v_k=k!\gamma_k\qquad \mbox{for}\qquad k=0,1,\ldots,n.
\end{equation}
It is well-known that all the coefficients $\gamma_k$ $(0\leq k\leq 2n)$ can be 
computed using fast polynomial multiplication---via the fast Fourier 
transform and its inverse---which has computational complexity $O(n\log{n})$. 
See, e.g., \cite[\S30]{Cormen}. This completes the proof.
\end{proof}

\begin{corollary}\label{C:Corollary1}
Theorem~\ref{T:MainThm} implies that the B\'{e}zier curve subdivision problem 
can be solved in $O(dM(n))$ time, where $M(n)$ denotes the cost of
multiplying two polynomials of degree $n$.
\end{corollary}

Let us point out interesting properties of the approach used to prove  
Theorem~\ref{T:MainThm}.

\begin{remark}\label{R:Remark1}
Observe that in the general case of $d>1$, it suffices to perform the FFT once 
for the input related to the polynomial $b_n$ (see~\eqref{E:Def-a_n-b_n}) as
it does not depend on the control points. The result can then be used for all 
dimensions.
\end{remark}

\begin{remark}\label{R:Remark2}
Again, without loss of generality, we assume that $d=1$. For a fixed natural 
$n$, suppose we have just computed all the coefficients $\gamma_k$ for 
$k=0,1,\ldots,2n$ (cf.~\eqref{E:Def-g_n-poly}, \eqref{E:Def-gamma_k}) using 
the FFT, the IFFT and the method described in the proof of
Theorem~\ref{T:MainThm}. Note that, in this case, $\gamma_{n+1},\gamma_{n+2},
\ldots,\gamma_{2n}$ are unused. So, if an additional non-zero $w_{n+1}$ is 
introduced, the values $v_0,v_1,\ldots,v_n$ (see Eq.~\eqref{E:Def-v_k}) remain 
unchanged, and most of the elements in $v_{n+1}$ are already computed in 
$\gamma_{n+1}$. It suffices to evaluate
$$
v_{n+1}=w_0(1-c)^{n+1}+(n+1)!\gamma_{n+1}+w_{n+1}c^{n+1}
$$
in $O(1)$ time, or, for $d$-dimensional control points, in $O(d)$ time.

This can be extended further, as we have
$$
v_{n+j}=(n+j)!\sum_{k=0}^{j-1}\alpha_k\beta_{n+j-k}+(n+j)!\gamma_{n+j}
         +(n+j)!\sum_{k=1}^{j} \alpha_{n+k} \beta_{j-k}\qquad (1\leq j\leq n).
$$
\end{remark}

\begin{corollary}\label{C:Corollary2}
When increasing $n$ by $m$ $(1\leq m\leq n)$, using the method described in
Remark~\ref{R:Remark2}, one gets the subdivision in $O(dm^2)$ additional time,
i.e., $O(d(n\log{n}+m^2))$ in total. This may be of interest in practice if the 
FFT implementation of choice is optimized for certain values of $n$.
\end{corollary}

\section{The practical solution}                          \label{S:P-Solution}

The new approach proposed in the previous section, despite its theoretical 
value, has limited practical use when applied directly. In numerical practice, 
it becomes unstable quickly and therefore is suitable mostly for low values of 
$n$.

There are probably two main reasons for such behavior of the method: 
\textit{(i)} the \textit{input signals} for the FFT, i.e., the 
coefficients~\eqref{E:Def-alpha_i-beta_i}, vanish quite rapidly for higher $i$,
$j$; \textit{(ii)} the \textit{output signals} of the IFFT, i.e., the
coefficients~\eqref{E:Def-gamma_k} (cf.~\eqref{E:Def-g_n-poly}), are quite 
small and have to be multiplied by a factorial in the final stage 
(cf.~\eqref{E:FinalStep}). 

Fortunately, the method can be slightly modified in a way that may mitigate 
this numerical drawback in some cases. To this end, we introduce a 
\textit{scaling factor}, which can amplify the input signals to the FFT and 
help overcome the numerical limitations discussed above. In practice, however, 
the numerical accuracy of this approach also depends on the particular 
implementations of the FFT and IFFT. Numerical experiments indicate that the 
scaled variant of the method performs well for $n\lesssim 50$ (see 
Section~\ref{S:Examples}).

As in the proof of Theorem~\ref{T:MainThm}, we consider 
Problem~\ref{P:MainProblem} in the case of $d=1$. If $d>1$, one simply applies 
the component-wise strategy. 

Now, given the scaling factor $s\neq 0$, and $c\in(0,1)$, $n\in\mathbb N$,
$w_i\in\mathbb R$ $(0\leq i\leq n)$, we rewrite~\eqref{E:Def-v_k} 
as
\begin{equation}\label{E:Def-v_k-scal}
v_k\equiv v_k(c)=\frac{k!}{s^k}
                    \sum_{i=0}^k\frac{s^iw_ic^i}{i!}\cdot
                                 \frac{s^{k-i}(1-c)^{k-i}}{(k-i)!}, 
\end{equation}
where $k=0,1,\ldots,n$. 

Consequently, we define
\begin{eqnarray}
\nonumber
&&a^s_n(x):=\sum_{i=0}^{n}\alpha^s_ix^i,\qquad
             \alpha^s_i:=s^i\alpha_i\qquad (0\leq i\leq n),\\[1ex]
&&
\label{E:Def-b^s_n}
b^s_n(x):=\sum_{j=0}^{n}\beta^s_jx^j,\qquad
             \beta^s_j:=s^j\beta_j\qquad (0\leq j\leq n)
\end{eqnarray}
(cf.~\eqref{E:Def-alpha_i-beta_i}), and $g^s_n(x):=a^s_n(x)b^s_n(x)$.

So, we have
$$
g^s_n(x)=\sum_{k=0}^{2n}\gamma^s_kx^k,
$$
where
\begin{equation}\label{E:Def-gamma^s_k}
\gamma^s_k:=\sum_{i=0}^{k}\alpha^s_i\cdot
                          \beta^s_{k-i}\qquad (0\leq k\leq 2n),
\end{equation}
and we set $\alpha^s_\ell:=0$, $\beta^s_\ell:=0$ if $\ell>n$.

Finally, for $k=0,1,\ldots,n$, we obtain
$$
v_k=s^{-k}k!\gamma^s_k.
$$
Certainly, if $s=1$, nothing has changed compared to 
Section~\ref{S:T-Solution}.

As previously recalled, using the FFT and the IFFT, all the coefficients 
$\gamma^s_0,\gamma^s_1,\ldots,\gamma^s_{2n}$ can be computed in $O(n\log{n})$
time. This implies that the required quantities $v_0, v_1,\ldots,v_n$ can also 
be found with the same computational complexity. 

It is worth noting that similar observations to those given in
Remark~\ref{R:Remark1} and Remark~\ref{R:Remark2} also apply in the 
\textit{scaling case}.

\begin{algorithm}[ht!]
\caption{Finding the left segment of the B\'{e}zier curve~\eqref{E:BezierCurve}
         for $c\in(0,1)$ using the scaling factor $s\neq0$ via the FFT and 
                                                       the IFFT.}\label{A:Alg1}
\begin{algorithmic}[1]
\Procedure {FastSubdivisionBC-L}{$c, s, [\p{W}_0, \p{W}_1, \ldots, \p{W}_n]$}
\State $c_1 \gets 1 - c$
\State $b_0 \gets 1$
\State $a_0 \gets 1$
\State $f_0 \gets 1$
\For {$i \gets 1, \ldots, n$}
  \State $a_i \gets a_{i-1} \cdot s \cdot c / i$
  \State $b_i \gets b_{i-1} \cdot s \cdot c_1 / i$
  \State $f_i \gets f_{i-1} \cdot i / s$
\EndFor
\State $B\_FFT \gets \mbox{\textsf{RFFT}}([b_0,b_1,\ldots,b_n], 2n+1)$
                     \Comment{Real FFT evaluating at $2n+1$ points.}
\For {$k \gets 1, \ldots, d$} 
                      \Comment{$d$ -- the dimension of control points.}
  \For {$i \gets 0, \ldots, n$} 
                      \Comment{$n+1$ -- the number of control points.}
    \State $w_i \gets a_i \cdot \p{W}_i[k]$
  \EndFor
  \State $A\_FFT \gets \mbox{\textsf{RFFT}}([w_0,w_1,\ldots,w_n], 2n+1)$  
  \State $C\_FFT \gets A\_FFT \cdot B\_FFT$
                     \Comment{Elementwise multiplication.}
  \State $X \gets \mbox{\textsf{IRFFT}}(C\_FFT, 2n+1)$
                     \Comment{Real IFFT evaluating at $2n+1$ points.}
  \For {$i \gets 0, \ldots, n$}  
       \State $result[k,i] \gets f_i\cdot X[i]$
  \EndFor  
\EndFor
\State \Return $result$
       \Comment{$\p{V}_i\equiv[result[1,i],result[2,i],\ldots,result[d,i]]$.}

\EndProcedure
\end{algorithmic}
\end{algorithm}

The new scaling factor-based approach for solving Problem~\ref{P:MainProblem} 
in its general form (i.e., for $d\in\mathbb N$) in $O(dn\log{n})$ time, using
Remark~\ref{R:Remark1}, is summarized in Algorithm~\ref{A:Alg1}. To optimize 
the method, we propose to use \textit{real} versions of FFT and IFFT (see, 
e.g., \cite{SJHB1987-corr}, \cite{SJHB1987}), as the \textit{input} and 
\textit{output} control points (cf.~\eqref{E:BezierCurve}, 
\eqref{E:Subdivision}, \eqref{E:PointsV}) have real coordinates. 

Examples illustrating the behavior of the new scaling factor-based method in
\textsf{Python} for a \textit{good choice} of $s$ are given in 
Section~\ref{S:Examples}.

\begin{remark}\label{R:LowDegrees}
For low degrees, it can be more efficient to find the control points of 
a fragment of a B\'{e}zier curve by directly using Eq.~\eqref{E:PointsV}. For 
example, for $n=2$, one would have to evaluate
$$
\p{V}_0\equiv\p{W}_0,\qquad 
\p{V}_1=(1-c)\p{W}_0+c\p{W}_1,\qquad 
\p{V}_2=(1-c)^2\p{W}_0+ 2c(1-c)\p{W}_1+c^2\p{W}_2,
$$
which, if the computations are organized well, can be done using $8d+2$ 
arithmetic operations instead of $9d+1$ for the de Casteljau algorithm. 
Similarly, for cubic curves, one can reduce the cost from $18d+1$ operations to 
$14d+5$.
\end{remark}

\begin{algorithm}[ht!]
\caption{Finding both segments of the B\'{e}zier curve~\eqref{E:BezierCurve}
         for $c\in(0,1)$ using the scaling factor $s\neq0$ via the FFT and
                                                       the IFFT.}\label{A:Alg2}
\begin{algorithmic}[1]
\Procedure {FastSubdivisionBC-BOTH}{$c,s,[\p{W}_0, \p{W}_1, \ldots, \p{W}_n]$}
\State $c_1 \gets 1 - c$
\State $b_0 \gets 1$
\State $a_0 \gets 1$
\State $f_0 \gets 1$
\For {$i \gets 1, \ldots, n$}
  \State $a_i \gets a_{i-1} \cdot s \cdot c / i$
  \State $b_i \gets b_{i-1} \cdot s \cdot c_1 / i$
  \State $f_i \gets f_{i-1} \cdot i / s$
\EndFor
\State $B\_FFT\_L \gets \mbox{\textsf{RFFT}}([b_0,b_1,\ldots,b_n], 2n+1)$
                     \Comment{Real FFT evaluating at $2n+1$ points.}
\State $B\_FFT\_R \gets \mbox{\textsf{RFFT}}([a_0,a_1,\ldots,a_n], 2n+1)$
                     \Comment{Real FFT evaluating at $2n+1$ points.}
\For {$k \gets 1, \ldots, d$}
                      \Comment{$d$ -- the dimension of control points.}
  \For {$i \gets 0, \ldots, n$}
                      \Comment{$n+1$ -- the number of control points.}
    \State $w_i \gets a_i \cdot \p{W}_i[k]$
    \State $z_i \gets b_i \cdot \p{W}_{n-i}[k]$
  \EndFor
  \State $A\_FFT\_L \gets \mbox{\textsf{RFFT}}([w_0,w_1,\ldots,w_n], 2n+1)$
  \State $A\_FFT\_R \gets \mbox{\textsf{RFFT}}([z_0,z_1,\ldots,z_n], 2n+1)$
  \State $C\_FFT\_L \gets A\_FFT\_L \cdot B\_FFT\_L$
                     \Comment{Elementwise multiplication.}
  \State $C\_FFT\_R \gets A\_FFT\_R \cdot B\_FFT\_R$
                     \Comment{Elementwise multiplication.}
  \State $X\_L \gets \mbox{\textsf{IRFFT}}(C\_FFT\_L, 2n+1)$
                     \Comment{Real IFFT evaluating at $2n+1$ points.}
  \State $X\_R \gets \mbox{\textsf{IRFFT}}(C\_FFT\_R, 2n+1)$
                     \Comment{Real IFFT evaluating at $2n+1$ points.}
  \For {$i \gets 0, \ldots, n$}  
       \State $result\_L[k,i] \gets f_i\cdot X\_L[i]$
       \State $result\_R[k,i] \gets f_i\cdot X\_R[i]$
  \EndFor                                            
\EndFor
\State \Return $result\_L$,  $result\_R$ 
             \Comment{$\p{V}^L_i\equiv[result\_L[1,i],result\_L[2,i],\ldots,
                                                            result\_L[d,i]]$,}
\Statex      \Comment{$\p{V}^R_i\equiv[result\_R[1,i],result\_R[2,i],\ldots,
                                                            result\_R[d,i]]$.}
\EndProcedure
\end{algorithmic}
\end{algorithm}

\begin{remark}\label{R:LeftRightSub}
Observe that the problem of finding the \textit{right part} of the B\'{e}zier 
curve~\eqref{E:BezierCurve}, i.e., the B\'{e}zier curve 
$\p{P}^{R}_n([0,1]):=\p{P}_n([c,1])$ $(0<c<1)$, is equivalent to finding the 
left part of a B\'{e}zier curve of the same degree $n$, with 
the control points $\p{W}_{n}, \p{W}_{n-1},\ldots,\p{W}_0$ and subdivision 
parameter $1-c$. Algorithm~\ref{A:Alg2} computes the control points for both 
the left and right subdivisions simultaneously.
\end{remark}

\subsection{The scaling factor}                       \label{SS:ScalingFactor}

Choosing an appropriate scaling factor $s\neq 0$ is not straightforward. Its 
value may depend on the degree of the curve, the positions of the control 
points, the subdivision parameter, the FFT and IFFT implementations used, and 
even the precision of floating-point arithmetic.

In Examples~\ref{E:Example1}, \ref{E:Example2} and \ref{E:Example3}, we use a 
scaling factor of the following simple form:
\begin{equation}\label{E:ScalingFactor}
s\equiv s(n):=0.375n+0.9,
\end{equation}
where $n$ is the degree of the B\'{e}zier curve being subdivided. This choice 
is based on extensive numerical experiments and provides satisfactory results 
in practice for $n\lesssim 50$. 

Let us analyze the possible choice of the scaling factor. For simplicity, 
suppose that $d=1$. In the scaled version ($s\neq 0$), we have
$$
v_k=s^{-k}k!\gamma_k^s \qquad (0\leq k\leq n).
$$
Assume that, when performed in floating-point arithmetic, the method yields
$$
\widetilde{\gamma}_k^s=\gamma_k^s+\delta_k^s\qquad (0\leq k\leq n),
$$
where $\delta_k^s$ denotes the absolute error in the computed value of 
$\gamma_k^s$. Then
$$
\left|v_k-\widetilde{v}_k^s\right|=|\delta_k^s|\frac{k!}{|s|^k},
$$
where $\widetilde{v}_k^s:=s^{-k}k!\widetilde{\gamma}_k^s$.

We observe that, in the unscaled version (i.e., $s=1$), even a small absolute 
error may be amplified by a factor of $k!$. To mitigate this effect, the 
scaling factor $s$ can be chosen so that the quantities $s^{-k}k!$ remain 
sufficiently small for all $k=0,1,\ldots,n$. However, the errors $\delta_k^s$ 
also depend on $s$, so minimizing the factors $s^{-k}k!$ alone does not 
necessarily minimize the overall numerical error. 

One possible choice is
$$
s:=(n!)^{1/n},
$$
obtained by solving the equation $s^{-n}n!=1$. Note that the linear function 
defined in~\eqref{E:ScalingFactor} provides a very good approximation to this 
quantity. Unfortunately, this approach to choosing the scaling factor fails for 
$n\gtrsim 50$, as it does not work well for all $k=0,1,\ldots,n$.

A potentially better choice for a restricted range of indices 
$k\in[n_1,n_2]\subseteq[0,n]$, where $n_1<n_2$, is a scaling factor $s$ 
satisfying
$$
\frac{n_1!}{s^{n_1}}=\frac{n_2!}{s^{n_2}}.
$$
This choice balances the rescaling coefficients at the endpoints of the 
interval and may improve the numerical accuracy of the computed values 
$v_{n_1},v_{n_1+1},\ldots,v_{n_2}$, particularly when the interval is not too 
wide. Using the well-known Stirling formula,
$$
\ell!\approx \sqrt{2\pi \ell}\left(\frac{\ell}{e}\right)^\ell
                                                     \qquad (\ell\in\mathbb N)
$$
(cf., e.g., \cite[(3.17)]{Cormen}), we obtain
\begin{equation}\label{E:MultiScalingFactor}
s\approx\frac{n_2}{e}\left(\frac{n_2}{n_1}\right)^{\frac{2n_1+1}{2(n_2-n_1)}}
\qquad (n_1>0).
\end{equation}

This heuristic observation suggests that a \textit{multiscale approach} 
(cf.~Example~\ref{E:Example4}) might be effective even for large $n$, although 
such a strategy would require repeated applications of the FFT and IFFT. 
Designing such a method appears to be challenging, and we leave its development 
for future research, particularly in view of the additional difficulties 
discussed in what follows.

Additional limitations arise from the finite precision of floating-point 
arithmetic. Below, we consider the case in which the curve is subdivided at 
its midpoint, i.e., for $c:=\frac{1}{2}$, which is important in 
many applications (see, e.g.,~\cite{Gravesen1997}, \cite{LR1980}).

We again consider the one-dimensional case, i.e., $d=1$. Suppose that $s>0$. 
For $c:=\frac{1}{2}$, we have
$$
\alpha^s_i=w_i\dfrac{s^i}{2^i i!},\qquad \beta^s_i=\dfrac{s^i}{2^i i!},
$$
where $i=0,1,\ldots,n$. Defining
$$
\lambda^s_i:=\dfrac{i!}{s^i}\qquad (0\leq i\leq n),
$$
we obtain $v_k=\lambda^s_k\gamma^s_k$ (cf.~\eqref{E:Def-v_k-scal}, 
\eqref{E:Def-gamma^s_k}) for $k=0,1,\ldots,n$.

Let us assume that $w_i\in[1,2]$ (see the examples in~\S\ref{S:Examples}). Then 
$\beta^s_i\leq\alpha^s_i\leq 2\beta^s_i$.

In floating-point arithmetic, the scaling factor should be chosen to avoid both 
overflow and underflow.

Let us examine $\beta^s_i$. Since
$$
\frac{\beta^s_i}{\beta^s_{i-1}}=\frac{s}{2i},
$$
we see that $\beta^s_i$ increases for $i<s/2$ and decreases for $i>s/2$. 
Similarly, since
$$
\frac{\lambda^s_i}{\lambda^s_{i-1}}=\frac{i}{s},
$$
we conclude that $\lambda^s_i$ decreases for $i<s$ and increases for $i>s$.

Thus, assuming that $1\leq s\leq n$, we may estimate a lower bound $m$ for 
the coefficients by considering $\lambda^s_s$ and $\beta^s_n$, and an upper 
bound $M$ by considering $\lambda^s_n$ and $2\beta^s_{s/2}$. Using again the 
Stirling formula, we obtain
$$
m\approx\min\left(\frac{\sqrt{2 \pi s}}{e^s},
                  \dfrac{1}{\sqrt{2\pi n}}\left(\dfrac{se}{2n}\right)^n
             \right)
$$
and
$$
M\approx\max\left(\sqrt{2\pi n}\left(\dfrac{n}{se}\right)^n,
                  \dfrac{2e^{s/2}}{\sqrt{\pi s}} 
             \right).
$$
In both cases, one of the two quantities depends only on $s$, and not 
explicitly on $n$. This allows us to derive an upper bound on $s$.

For \texttt{float64} arithmetic, consider the condition
$$
2^{-1020}\leq\sqrt{2\pi s}e^{-s},
$$
which is equivalent to
$$
s\leq 1020\log 2+\frac12\log(2\pi s).
$$
The largest positive integer $s$ satisfying this inequality is $711$.

Similarly, to avoid overflow, we impose
$$
2^{1020}\geq e^{s/2}\frac{2}{\sqrt{\pi s}},
$$
which yields
$$
s\leq \log(\pi s)+2038\log{2}.
$$
This gives a less restrictive upper bound of $1421$ for integer values of $s$.

Thus, it suffices to ensure that 
$\dfrac{1}{\sqrt{2\pi n}}\left(\dfrac{se}{2n}\right)^n$ does not underflow and 
that $\sqrt{2\pi n}\left(\dfrac{n}{se}\right)^n$ does not overflow.

To avoid underflow, we require
$$
2^{-1020}\leq\left(\dfrac{se}{2n}\right)^n\frac{1}{\sqrt{2\pi n}},
$$
which is equivalent to
$$
2^{-1020/n}(2\pi n)^{1/(2n)}n\frac{2}{e}\leq s.
$$
The expression on the left-hand side is non-decreasing as a function of $n$. 
For $s=711$, the largest admissible integer $n$ is $1529$.

Similarly, to avoid overflow, we require
$$
2^{1020}\geq\left(\dfrac{n}{se}\right)^n\sqrt{2\pi n},
$$
which yields
$$
2^{-1020/n}\frac{n}{e}(2\pi n)^{1/(2n)}\leq s.
$$
Proceeding analogously, we obtain $2546$ as the largest admissible integer $n$.

One can also verify that
$$
2^{-1020/n}(2\pi n)^{1/(2n)}n\frac{2}{e}
\geq
2^{-1020/n}\frac{n}{e}(2\pi n)^{1/(2n)}.
$$
Thus, to avoid both underflow and overflow, for a given $n$, we need to choose
$$
s\in\left[2^{-1020/n}(2\pi n)^{1/(2n)}\frac{2n}{e},711\right].
$$

For sufficiently large $n$, the interval becomes empty. More precisely, it is 
non-empty if and only if $1\leq n\leq 1529$. Consequently, for the scaling 
strategy considered here, the degree of the Bézier curve imposes restrictions 
on the range of floating-point numbers required to avoid underflow and 
overflow.

\section{Other applications}                                \label{S:OtherApp}

In this section, we briefly outline additional possible applications of the new 
method for subdividing polynomial B\'{e}zier curves. More precisely, we 
consider: \textit{(i)} computing the derivatives of polynomial B\'{e}zier 
curves; \textit{(ii)} subdividing rational B\'{e}zier curves; and 
\textit{(iii)} subdividing rectangular B\'{e}zier surfaces. A detailed 
investigation of these issues, particularly in terms of practical applicability 
and computational performance, is left for future research.


\subsection{Derivatives of polynomial B\'{e}zier curves at the endpoints} 
                                                             \label{SS:DiffBC}

In many applications, it is necessary to compute the derivatives of B\'{e}zier 
curves. For example, to find the vectors
\begin{equation}\label{E:BezierDerivatives}
P^{(1)}_n(t), P^{(2)}_n(t),\ldots, P_n^{(r)}(t)\in\mathbb R^d
\end{equation}
(cf.~\eqref{E:BezierCurve}) for a given $t\in[0,1]$ and $1\leq r\leq n$, we can
also use the de Casteljau algorithm~\eqref{E:DC-alg}. In that case, the 
computational complexity is $O(d(n^2 + r^2))$. New, faster methods achieving 
$O(drn)$ time complexity have recently been proposed by the authors. 
For details, see~\cite{ChW2024}, where the problem of computing the derivatives
of rational B\'{e}zier curves is considered as well.

It seems that the use of the FFT and the IFFT allows to accelerate the new 
methods proposed in the mentioned article. In particular, Algorithm 2.2 
from~\cite{ChW2024}, which computes all the vectors~\eqref{E:BezierDerivatives} 
in $O(drn)$ time, may be optimized. However, the modification is rather 
complex, and we leave it for future research.

For this reason, we restrict ourselves to the problem of 
computing all the vectors~\eqref{E:BezierDerivatives} at the ends of a curve
(i.e., for $t\in\{0,1\}$) which is of great importance, for example, when 
joining a sequence of polynomial B\'{e}zier curves smoothly.

It is simple to check that
\begin{eqnarray}
&&\label{E:Derivative0}
P_n^{(k)}(0)=\dfrac{n!}{(n-k)!}\sum_{i=0}^k\binom{k}{i}(-1)^{k-i}\p{W}_i,\\
&&\label{E:Derivative1}
P_n^{(k)}(1)=\dfrac{n!}{(n-k)!}\sum_{i=n-k}^n\binom{k}{n-i}(-1)^{n-i}\p{W}_i,
\end{eqnarray}
where $k=0,1,\ldots,n$. See, e.g., \cite[Eq.~(5.25) and (5.26)]{Farin2002}.
Let us focus on~\eqref{E:Derivative0}, as the method for~\eqref{E:Derivative1} 
can be derived similarly.

Certainly, the vectors $P_n^{(k)}(0)\in\mathbb R^d$ $(1\leq k\leq n)$ can be
derived component-wise. So, let us assume $d=1$. Observe that to evaluate all 
the sums
$$
\sum_{i=0}^k\binom{k}{i}z_i\qquad (z_i\in\mathbb R) 
$$
for all $k=0,1,\ldots,n$, where $n\in\mathbb N$ is fixed, it is possible to
apply the approach from Section~\ref{S:P-Solution} (cf.~\eqref{E:Def-v_k-scal}; 
see also Section~\ref{S:T-Solution} and~\eqref{E:Def-v_k}). Indeed, it is 
enough to set $c:=\tfrac{1}{2}$ there, because 
$B^k_i(\tfrac{1}{2})=2^{-k}\binom{k}{i}$ $(0\leq i\leq k)$.

To sum up, one can compute all the vectors 
$P_n^{(k)}(0), P_n^{(k)}(1)\in\mathbb R^d$ for all $1\leq k\leq n$ with 
computational complexity $O(dn\log{n})$. This approach has lower asymptotic 
computational complexity than the de Casteljau algorithm or the methods 
from~\cite{ChW2024}, all of which run in $O(dn^2)$ time. 

\subsection{Subdivision of rational B\'{e}zier curves}        \label{SS:RatBC}

Let us consider a \textit{rational B\'{e}zier curve 
$\p{R}_n:[0,1]\rightarrow{\mathbb E}^d$ $(d\in\mathbb N)$ of degree 
$n\in\mathbb N$} with the \textit{control points} 
$\p{W}_0,\p{W}_1,\ldots,\p{W}_n\in{\mathbb E}^d$, and the corresponding 
\textit{weights} $\omega_0,\omega_1,\ldots,\omega_n>0$, 
\begin{equation*}
\p{R}_n(t):=\frac{\displaystyle\sum_{i=0}^{n}\omega_iB^n_i(t)\p{W}_i}
                 {\displaystyle \sum_{i=0}^{n}\omega_iB^n_i(t)}
                                                       \qquad (0\leq t\leq 1).
\end{equation*}

A point on the parametric curve $\p{R}_n$ can be evaluated using the 
\textit{rational de Casteljau algorithm}, which has a computational complexity 
$O(dn^2)$. Note that the rational de Casteljau algorithm is significantly more 
expensive compared to its polynomial version, as it requires $O(n^2)$ 
divisions. Similarly to the polynomial case, this method also subdivides a 
rational B\'{e}zier curve into two segments in $O(dn^2)$ time.

For example, for a given $c\in(0,1)$, the \textit{left part} of the curve 
$\p{R}_n$, i.e., the rational B\'{e}zier curve 
$\p{R}^L_n([0,1]):=\p{R}_n([0,c])$, has the following form:
\begin{equation*}
\p{R}^L_n(t)=\dfrac{\displaystyle\sum_{i=0}^{n}\nu_iB^n_i(t)\p{V}_i}
                    {\displaystyle\sum_{i=0}^{n}\nu_iB^n_i(t)}
                                                       \qquad (0\leq t\leq 1),
\end{equation*}
where
\begin{equation}\label{E:Def-nu_k-p_V}
\nu_k\equiv\nu_k(c):=\sum_{i=0}^{k}\omega_iB^k_i(c),
\qquad 
\p{V}_k\equiv\p{V}_k(c):=\frac{1}{\nu_k}\sum_{i=0}^{k}\omega_iB^k_i(c)\p{W}_i
\qquad (0\leq k\leq n).
\end{equation}
As already mentioned, all these weights and points can be found via the 
rational de Casteljau algorithm in $O(dn^2)$ time. 

For a justification of the given facts, further information on rational 
B\'{e}zier curves and the rational de Casteljau algorithm, see, e.g., 
\cite[\S13]{Farin2002}.

It should be noted that the new geometric linear-time method~\cite{WCh2020},
derived by the authors, allows for faster computation of a point on a rational 
B\'{e}zier curve, namely in $O(dn)$ time. However, subdividing a rational 
B\'{e}zier curve into two parts using this new method still requires $O(dn^2)$ 
arithmetic operations. See~\cite[p.~59]{FChPhD}.

It is not difficult to see that all the weights and the control
points~\eqref{E:Def-nu_k-p_V} can be computed with computational complexity 
$O(dn\log{n})$. Indeed, we first need to determine $\nu_0,\nu_1,\ldots,\nu_n$ 
using the new method proposed in Section~\ref{S:P-Solution} 
(cf.~\eqref{E:Def-v_k-scal}; see also Section~\ref{S:T-Solution} 
and~\eqref{E:Def-v_k}). Next, we can compute $\p{V}_0,\p{V}_1,\ldots,\p{V}_n$ 
in a component-wise manner using the new approach again.

\subsection{Subdivision of rectangular B\'{e}zier surfaces}   \label{SS:RetBS}

The \textit{rectangular B\'{e}zier surface} 
$\p{S}_{nm}:[0,1]^2\rightarrow\mathbb E^d$ $(d\in\mathbb N)$ \textit{of degree} 
$(n,m)\in\mathbb N^2$ with the \textit{control points} 
$\p{W}_{ij}\in{\mathbb E}^d$ $(0\leq i\leq n,\; 0\leq j\leq m)$ is defined by
$$
\p{S}_{nm}(t,u):=\sum_{i=0}^n\sum_{j=0}^mB^n_i(t)B^m_j(u)\p{W}_{ij}
                                                       \qquad(0\leq t,u\leq1).
$$
To evaluate a point on a rectangular B\'{e}zier surface, one can use the de 
Casteljau algorithm multiple times. The computational complexity in this case 
is $O(d\overline{M}(\underline{M}^2+\overline{M}))$, where 
$\underline{M}:=\min(n,m)$, $\overline{M}:=\max(n,m)$. See, e.g., 
\cite[\S14]{Farin2002}.  

Recently, the authors have proposed faster methods of computing a point on 
the surface $\p{S}_{nm}$, which work in $O(dnm)$ time. See~\cite{WCh2020} and
\cite[\S2.4]{FChPhD}.

The problem of subdividing rectangular B\'{e}zier surfaces is important in 
practical applications, and for this reason, it has been considered in the 
literature. Note that the de Casteljau algorithm is also helpful in solving 
this task. For details, see, e.g., \cite[\S14.7 and \S16.1]{Farin2002}.

Let us restrict ourselves only to the following case of rectangular B\'{e}zier 
surface subdivision. Consider the rectangular B\'{e}zier surface 
$\p{S}^L_{nm}$ representing the \textit{left part} of the patch $\p{S}_{nm}$, 
i.e., $\p{S}^L_{nm}([0,1],[0,1]):=\p{S}_{nm}([0,c],[0,1])$, where $c\in(0,1)$ 
is fixed. 

It can be checked that
$$
\p{S}^{L}_{nm}(t,u)=\sum_{i=0}^n\sum_{j=0}^mB^n_i(t)B^m_j(u)\p{V}_{ij}
                                                       \qquad(0\leq t,u\leq1),
$$
where the points $\p{V}_{ij}\in\mathbb E^d$ are given by
\begin{equation}\label{E:Def-V_ij}
\p{V}_{ij}\equiv\p{V}_{ij}(c):=\sum_{k=0}^iB^i_k(c)\p{W}_{kj}
                                    \qquad (0\leq i\leq n,\; 0\leq j\leq m).
\end{equation}
These quantities can be computed using the de Casteljau algorithm $m+1$ 
times in $O(dmn^2)$ time (cf.~Section~\ref{S:Problem}). 

However, it is now clear that the new method from Section~\ref{S:P-Solution} 
(cf.~\eqref{E:Def-v_k-scal}; see also Section~\ref{S:T-Solution} 
and~\eqref{E:Def-v_k}) allows us to derive all the points~\eqref{E:Def-V_ij} 
with the computational complexity $O(dmn\log{n})$. 

\begin{remark}\label{R:Remark3}
Using the new approach, one can take advantage of the fact that we need to 
compute $m+1$ sums of the form~\eqref{E:Def-V_ij} for different sets of control 
points which all correspond to the same value of $c$. This suggests that a 
notable improvement can be achieved. Indeed, it suffices to run the FFT once 
for the input related to the polynomial $b^s_n$ (see~\eqref{E:Def-b^s_n}; or 
to the polynomial $b_n$, cf.~\eqref{E:Def-a_n-b_n}), and the result can then be 
applied to all dimensions (cf.~Remark~\ref{R:Remark1}) and for all 
$\ell=0,1,\ldots,m$.
\end{remark}

Note that to obtain one quarter of a rectangular B\'{e}zier patch, the method 
outlined above must be applied twice, whereas subdividing the patch into four 
subpatches using the new method requires six applications.

By combining the approaches described in this and the preceding paragraph, a 
similar method for asymptotically faster subdivision of \textit{rational 
rectangular B\'{e}zier surfaces} (see, e.g., \cite[\S16.6]{Farin2002}; 
cf.~also~\S\ref{SS:RatBC}) can also be developed.

\section{Numerical tests}                                   \label{S:Examples}

The following experiments have been performed. The results have been obtained 
on a~computer with \texttt{Intel Core i5-10600 CPU} at \texttt{3.30GHz} 
processor and \texttt{16GB} \texttt{RAM}, using \textsf{Python 3.14.4} 
(double precision) and \texttt{numpy.fft.rfft} (real FFT), 
\texttt{numpy.fft.irfft} (real IFFT) \textsf{NumPy 2.3.5} functions 
(for \textit{real} versions of FFT and IFFT, see, e.g., \cite{SJHB1987-corr},
\cite{SJHB1987}). The source code which was used to perform the tests is 
available at 
\url{https://github.com/filipchudy/bezier-subdivision-fft/tree/main}.

\begin{table*}[ht!]
\begin{center}
\renewcommand{\arraystretch}{1.15}
\begin{tabular}{lccc}
$n$ & de~Casteljau & Alg.~\ref{A:Alg1} & Eq.~\eqref{E:Def-gamma_k} \\ \hline
2 & \bf{3.43} & 11.08 & 5.14 \\ 
3 & \bf{4.64} & 11.45 & 5.95 \\ 
4 & \bf{5.89} & 11.69 & 6.77 \\ 
5 & \bf{7.17} & 11.96 & 7.54 \\ 
6 & 8.42 & 12.24& \bf{8.40} \\ 
7 & 9.67 & 12.21& \bf{9.16} \\ 
8 & 10.93 & 12.62 & \bf{10.00} \\ 
9 & 12.22 & 12.88 & \bf{10.76} \\ 
10 & 13.52 & 12.97 & \bf{11.53} \\ 
11 & 14.75 & 15.28 & \bf{12.32} \\ 
12 & 16.04 & 15.21 & \bf{13.23} \\ 
13 & 17.30 & 15.41 & \bf{13.93} \\ 
14 & 18.51 & 16.34 & \bf{14.76} \\ 
15 & 19.85 & 16.65 & \bf{15.52} \\ 
16 & 21.07 & 16.57 & \bf{16.31} \\ 
17 & 22.35 & \bf{16.83} & 17.05 \\ 
18 & 23.70 & \bf{17.66} & 17.89 \\ 
19 & 24.91 & \bf{17.45} & 18.65 \\ 
20 & 26.20 & \bf{18.51} & 19.48 \\ 
25 & 32.91 & \bf{18.52} & 26.09 \\ 
30 & 38.91 & \bf{22.14} & 30.64 \\ 
35 & 45.46 & \bf{24.48} & 34.99 \\ 
40 & 51.91 & \bf{21.83} & 39.46 \\ 
45 & 58.14 & \bf{23.85} & 44.13 \\ 
50 & 64.81 & \bf{31.90} & 48.69 \\ 
\end{tabular}
\renewcommand{\arraystretch}{1}
\vspace{2ex}
\caption{Running times comparison (in seconds) for Example~\ref{E:Example1}.}%
\label{T:Table2}
\vspace{-3ex}
\end{center}
\end{table*}

\begin{example}\label{E:Example1}
Table~\ref{T:Table2} shows the comparison between the running times of the de 
Casteljau algorithm, Algorithm~\ref{A:Alg1} with the scaling factor 
$s\equiv s(n):=0.375n+0.9$, where $n$ is a degree of a B\'{e}zier curve, and 
the method which directly evaluates~\eqref{E:Def-gamma_k}, for two-dimensional 
B\'{e}zier curves. Note that the proposed form of the scaling factor is based 
on extensive numerical testing (see also \S\ref{SS:ScalingFactor}). 

The following numerical experiments have been conducted. For fixed $n$, $1000$ 
test sets consisting of one curve of degree $n$ are generated. Their control 
points $\p{W}_k\in[1, 2]^2$ $(0\leq k\leq n)$ have been generated using the 
\texttt{numpy.random.uniform} \textsf{NumPy} function. Each curve is then 
subdivided at 499 points $c_i:=i/500$ $(1\leq i\leq 499)$. Each algorithm is 
tested using the same curves. Note that, for each curve, the $\beta_k^s$ 
sequence is reused for all dimensions of the curve 
(cf.~Remark~\ref{R:Remark1}). The $\beta_k^s$ sequence is not reused across 
different curves.  

Table~\ref{T:Table2} shows the total running time of all $499\times 1000$ 
subdivisions.
\end{example}

\begin{table*}[ht!]
\begin{center}
\renewcommand{\arraystretch}{1.15}
\begin{tabular}{lcccc}
$n$ &  min Alg.~\ref{A:Alg1} &  mean Alg.~\ref{A:Alg1} &  
       min Eq.~\eqref{E:Def-gamma_k} & mean Eq.~\eqref{E:Def-gamma_k} \\ \hline
2 & 14.95 & 	16.70 & 15.31 & 17.69 \\
3 & 14.52 & 	16.41 & 15.13 & 17.46 \\
4 & 13.88 & 	16.06 & 15.10 & 17.31 \\
5 & 13.27 & 	15.86 & 15.05 & 17.19 \\
6 & 12.45 & 	15.56 & 15.05 & 17.09 \\
7 & 11.62 & 	15.28 & 14.96 & 17.01 \\
8 & 10.72 & 	14.86 & 14.96 & 16.94 \\
9 & 9.79 & 14.39	 & 14.96 & 16.88 \\
10 & 8.72 & 	14.04 & 14.96 & 16.83 \\
11 & 13.86 & 15.75 & 14.88 & 16.79 \\
12 & 13.79 & 15.78 & 14.88 & 16.75 \\
13 & 13.62 & 15.75 & 14.88 & 16.71 \\
14 & 13.42 & 15.62 & 14.88 & 16.68 \\
15 & 13.20 & 15.56 & 14.88 & 16.64 \\
16 & 13.23 & 15.55 & 14.82 & 16.62 \\
17 & 13.07 & 15.46 & 14.82 & 16.59 \\
18 & 12.77 & 15.39 & 14.82 & 16.56 \\
19 & 12.77 & 15.36 & 14.82 & 16.54 \\
20 & 12.46 & 15.26 & 14.82 & 16.52 \\
25 & 11.79 & 14.91 & 14.74 & 16.38 \\
30 & 10.95 & 14.72 & 14.70 & 16.30 \\
35 & 10.18 & 14.32 & 14.66 & 16.24 \\
40 & 9.51 & 14.24 & 14.64 & 16.19 \\
45 & 8.80 & 13.91 & 14.59 & 16.14 \\
50 & 7.79 & 13.65 & 14.59 & 16.10 \\
\end{tabular}
\renewcommand{\arraystretch}{1}
\vspace{2ex}
\caption{The minimum and mean number of exact significant decimal digits for 
Example~\ref{E:Example2}.}\label{T:Table3}
\vspace{-3ex}
\end{center}
\end{table*}

\begin{example}\label{E:Example2}
The same tests as in Example~\ref{E:Example1} were carried out. For each 
coordinate of every resulting control point of the subdivided curve, the number 
of correct decimal digits was determined, treating the results of the de 
Casteljau algorithm as exact. Table~\ref{T:Table3} presents the aggregate 
statistics on numerical precision, i.e., the minimum and mean number of exact 
significant decimal digits relative to the results obtained by the de Casteljau 
algorithm.
\end{example}

The experiments show that the scaled version of the method is effective, 
particularly for curves of higher degree. It outperforms the traditional 
approach for curves of degree $n>9$. The minimum number of correct 
significant decimal digits decreases as $n$ increases, while the mean remains 
close to the precision of the floating-point arithmetic used.

The direct application of Eq.~\eqref{E:Def-gamma_k} yields very good numerical 
results, indicating that the numerical problems of the new method for 
higher-degree curves (cf.~Table~\ref{T:Table3}) are closely related to the 
quality of the FFT and IFFT implementations employed. It should also be noted 
that the method based on the direct application of Eq.~\eqref{E:Def-gamma_k} is 
faster than using the de Casteljau algorithm already for curves of degree 
$n>5$, although both approaches have the same asymptotic computational 
complexity.

While the de Casteljau algorithm performs better for degrees 4 and 5, applying 
Remark~\ref{R:LowDegrees} to quadratic and cubic B\'{e}zier curves results in 
shorter computation times than the de Casteljau algorithm. See 
Table~\ref{T:Table4}.

\begin{table*}[ht!]
\begin{center}
\renewcommand{\arraystretch}{1.15}
\begin{tabular}{lcc}
$n$ & de~Casteljau & using Rem.~\ref{R:LowDegrees} \\ \hline
2 & 3.43 & \bf{2.84} \\
3 & 4.64 & \bf{4.56} \\ 
\end{tabular}
\renewcommand{\arraystretch}{1}
\vspace{2ex}
\caption{Running times comparison (in seconds).}%
\label{T:Table4}
\vspace{-3ex}
\end{center}
\end{table*}

In summary, numerical experiments suggest that the following 
\textit{hybrid method} should be used in practice for $n\leq 50$ 
(cf.~\S\ref{SS:ScalingFactor}), i.e.,
\begin{itemize}
\itemsep 1ex

\item the approach described in Remark~\ref{R:LowDegrees} if $n=2,3$,

\item the de Casteljau algorithm if $n=4,5$,

\item Eq.~\eqref{E:Def-gamma_k} if $6\leq n\leq 16$, 

\item and Algorithm~\ref{A:Alg1} for $16<n\leq 50$.

\end{itemize}

As discussed above, the scaled version of the method proposed in this paper may 
lose numerical accuracy for $n\gtrsim 50$. However, the multiscale approach 
outlined in \S\ref{SS:ScalingFactor} appears promising. The following example 
illustrates its potential.

\begin{example}\label{E:Example4}
We tested a multiscale strategy for a one-dimensional B\'{e}zier curve
of degree $100$ (i.e., a polynomial of degree $100$ in Bernstein form),
whose control coefficients $w_k\in[1,2]$ $(0\leq k\leq 100)$ were
generated using the \textsf{NumPy} function \texttt{numpy.random.uniform}.

In the two-stage approach, the control coefficients $v_k$ of the left subcurve 
were computed as follows:
\begin{itemize}
\itemsep 1ex

\item For $k=0,1,\ldots,29$, we subdivided the curve of degree $29$
      with control coefficients $w_k$ $(0\leq k\leq 29)$, using the
      scaling factor $s:=s(29)$ (see~\eqref{E:ScalingFactor}).

\item For $k=30,31,\ldots,100$, we subdivided the original curve, using the 
      scaling factor given by formula~\eqref{E:MultiScalingFactor} with 
      $n_1=30$ and $n_2=100$.

\end{itemize}

In the three-stage approach, the ranges of indices and the corresponding 
scaling factors were:
\begin{itemize}
\itemsep 1ex

\item $k\in[0,14]$, $s=s(14)$ (cf.~\eqref{E:ScalingFactor}),

\item $k\in[15,47]$, with $s$ determined by 
      formula~\eqref{E:MultiScalingFactor} for $n_1=15$ and $n_2=47$,

\item $k\in[48,100]$, with $s$ determined by 
      formula~\eqref{E:MultiScalingFactor} for $n_1=48$ and $n_2=100$.
\end{itemize}

The boundaries $\{29\}$ and $\{14,47\}$, respectively, were found 
experimentally to maximize the global minimum number of matching decimal 
digits.

Figure~\ref{F:Figure1} shows, for each control coefficient of the left
subcurve, the minimum number of matching significant decimal digits
relative to the results obtained by the de Casteljau algorithm.
The minimum was taken over $1000$ randomly generated curves, each
subdivided at $499$ parameter values $c_i:=i/500$ $(1\leq i\leq499)$.
Note the dips in precision near the indices where the scaling
factors change.

Figure~\ref{F:Figure2} illustrates the accuracy of the three-stage
approach for curves of degree $200$. In this case, the interval
boundaries were doubled to preserve their relative positions,
so the scaling factors change at indices $28$ and $94$.
\end{example}

\begin{figure*}[ht!]
\centering
\includegraphics[angle=0,width=1\linewidth]{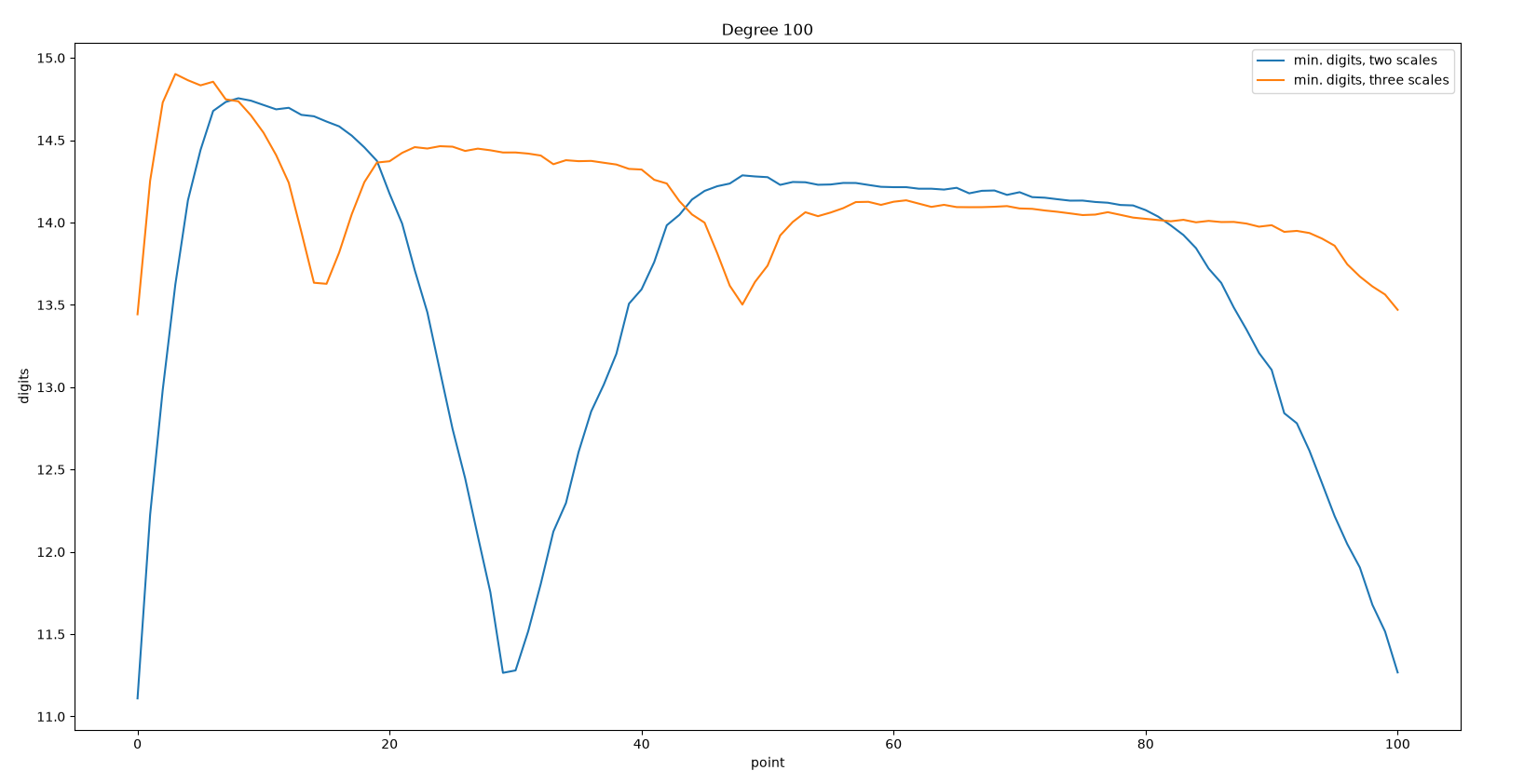}%

\caption{Numerical accuracy of the two-stage and three-stage approaches
for a one-dimensional B\'{e}zier curve of degree $100$.}\label{F:Figure1}
\end{figure*}

\begin{figure*}[ht!]
\centering
\includegraphics[angle=0,width=1\linewidth]{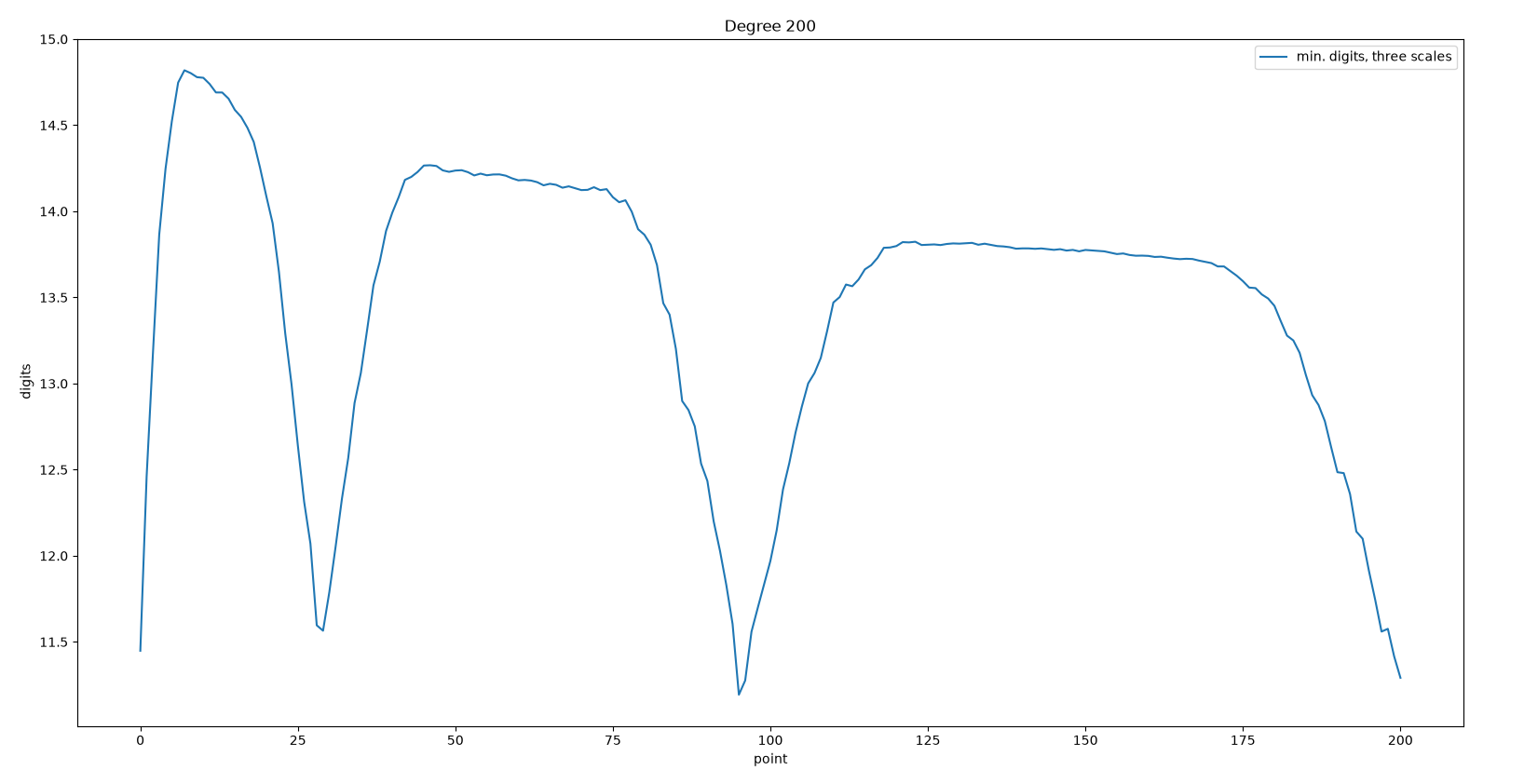}%

\caption{Numerical accuracy of the three-stage approach for a one-dimensional 
B\'{e}zier curve of degree $200$.}\label{F:Figure2}
\end{figure*}

We conclude this paragraph with an example related to the subdivision of 
rational B\'{e}zier curves (see Section~\ref{SS:RatBC}).

\begin{table*}[ht!]
\begin{center}
\renewcommand{\arraystretch}{1.15}
\begin{tabular}{lcc}
$n$ & rat.~de~Casteljau & new method\\ \hline
2 & 9.75  & \bf{8.27} \\
3 & 13.79 & \bf{12.10} \\
4 & 17.86 & \bf{16.78} \\ 
5 & 21.95 & \bf{17.08} \\ 
6 & 26.04 & \bf{18.90} \\ 
7 & 30.16 & \bf{20.62} \\ 
8 & 34.20 & \bf{22.52} \\ 
9 & 38.20 & \bf{24.28} \\ 
10 & 42.29 & \bf{26.16} \\
\end{tabular}
\renewcommand{\arraystretch}{1}
\vspace{2ex}
\caption{Running times comparison (in seconds) for Example~\ref{E:Example3}.}%
\label{T:Table5}
\vspace{-3ex}
\end{center}
\end{table*}

\begin{example}\label{E:Example3}
For the subdivision of a rational B\'{e}zier curve, the following experiments 
have been conducted. The control points have been sampled like 
in Example~\ref{E:Example1}. The weights of the points were sampled from 
a uniform distribution over $[0.01,1]$. Each of 1000 rational curves of degree 
$n$ is subdivided at 499 points $c_i:=i/500$ $(1 \leq i \leq 499)$. The 
compared algorithms were the rational de Casteljau algorithm and a new method 
based on~\eqref{E:Def-nu_k-p_V}, with~\eqref{E:Def-gamma_k} used to subdivide 
the polynomial B\'{e}zier curves in the numerator and the denominator. For 
quadratic, cubic and quartic rational B\'{e}zier curves, an unrolled version of 
the algorithm is used (cf.~Remark~\ref{R:LowDegrees}). Table~\ref{T:Table5} 
shows the total running time of all $499\times 1000$ subdivisions.
\end{example}

\section{Conclusions}                                    \label{S:Conclusions}

We have proposed a new method for subdividing a B\'{e}zier curve of degree $n$ 
in ${\mathbb E}^d$, with a computational complexity of $O(dn\log n)$. The 
algorithm uses the fast Fourier transform and its inverse. A scaled variant of 
the method can improve numerical accuracy in some cases, even for B\'{e}zier 
curves of relatively high degree. Based on our numerical experiments, we 
propose a \textit{hybrid method} that may often provide an effective tool for 
solving the B\'{e}zier curve subdivision problem if $n\lesssim 50$. 

The choice of an appropriate scaling factor remains an important issue, 
particularly for curves of higher degree. As discussed in 
Section~\ref{SS:ScalingFactor}, a possible way to improve numerical accuracy is 
to develop a multiscale variant of the FFT-based method, using different 
scaling factors for different ranges of control point indices. Such an approach 
would require multiple FFT and IFFT computations and deserves further 
investigation.

Furthermore, the proposed approach can be extended to yield asymptotically 
faster algorithms for evaluating B\'{e}zier curve derivatives, subdividing 
rational B\'{e}zier curves, and subdividing rectangular polynomial or rational 
B\'{e}zier patches. However, the practical computational performance and 
numerical accuracy of these extensions remain to be investigated.

Note that the numerical experiments presented in \S6 are intended to illustrate 
the main features and limitations of the proposed method rather than to provide 
an exhaustive assessment of its numerical performance, which may also depend on 
the particular FFT and IFFT implementations used. It should also be emphasized 
that many FFT and IFFT implementations are optimized for large input sizes 
(e.g., for applications in audio processing), which may also affect the 
performance of the new method.

Beyond CAGD, the proposed approach may also be useful in other methods and 
algorithms involving polynomial reparameterization.

\section*{Acknowledgments}                                    

The authors would like to thank the anonymous reviewers for their constructive 
criticism and many valuable comments and suggestions, which significantly 
helped improve the results and presentation of the paper.

\section*{Funding}                                    

Funding not applicable.

\bibliographystyle{elsart-num-sort}
\biboptions{sort&compress}
\bibliography{BC-fast-subdivision-NUMA-rev}

\end{document}